\documentclass[aps,prx,nofootinbib,twocolumn,superscriptaddress,10pt]{revtex4-2}
\usepackage{mathtools}
\usepackage{mathrsfs}
\usepackage[ruled, vlined, linesnumbered]{algorithm2e}
\usepackage[table]{xcolor}
\usepackage{bbm,bm}
\usepackage{amsmath,amsfonts,amssymb}
\usepackage{array,graphicx,multirow,tcolorbox,relsize,booktabs,dsfont}
\usepackage[utf8]{inputenc}
\usepackage[T1]{fontenc}
\usepackage{pifont}
\usepackage{url}
\usepackage{theorem}

\newtheorem{proposition}{Proposition}
\newtheorem{lemma}[proposition]{Lemma}

\newtheorem{theorem}[proposition]{Theorem}

\newenvironment{proof}{\noindent \textit{{Proof.}}}{\hfill $\square$}

\newcommand{\nc}{\newcommand}
\nc{\rnc}{\renewcommand}
\nc{\lbar}[1]{\overline{#1}}
\nc{\bra}[1]{\langle#1|}
\nc{\ket}[1]{|#1\rangle}
\nc{\dketbra}[2]{\vert #1 \rangle \hspace{-.8mm} \rangle \hspace{-.4mm} \langle\hspace{-.8mm}\langle #2 \vert}
\nc{\dbra}[1]{\langle\hspace{-.8mm}\langle #1\vert}
\nc{\dket}[1]{\vert#1\rangle\hspace{-.8mm}\rangle}
\nc{\ketbra}[2]{|#1\rangle\!\langle#2|}
\nc{\braket}[2]{\langle#1|#2\rangle}
\nc{\braandket}[3]{\langle #1|#2|#3\rangle}

\nc{\rank}{\operatorname{rank}}
\nc{\tr}{\operatorname{Tr}}
\nc{\ox}{\otimes}
\nc{\cA}{{\cal A}}
\nc{\cB}{{\cal B}}
\nc{\cC}{{\cal C}}
\nc{\cD}{{\cal D}}
\nc{\cE}{{\cal E}}
\nc{\cF}{{\cal F}}
\nc{\cG}{{\cal G}}
\nc{\cH}{{\cal H}}
\nc{\cI}{{\cal I}}
\nc{\cJ}{{\cal J}}
\nc{\cK}{{\cal K}}
\nc{\cL}{{\cal L}}
\nc{\cM}{{\cal M}}
\nc{\cN}{{\cal N}}
\nc{\cO}{{\cal O}}
\nc{\cP}{{\cal P}}
\nc{\cQ}{{\cal Q}}
\nc{\cR}{{\cal R}}
\nc{\cS}{{\cal S}}
\nc{\cT}{{\cal T}}
\nc{\cU}{{\cal U}}
\nc{\cV}{{\cal V}}
\nc{\cX}{{\cal X}}
\nc{\cY}{{\cal Y}}
\nc{\cZ}{{\cal Z}}
\nc{\cW}{{\cal W}}

\nc{\sfA}{{\mathsf A}}
\nc{\sfB}{{\mathsf B}}
\nc{\sfC}{{\mathsf C}}
\nc{\sfD}{{\mathsf D}}
\nc{\sfE}{{\mathsf E}}
\nc{\sfF}{{\mathsf F}}
\nc{\sfG}{{\mathsf G}}
\nc{\sfH}{{\mathsf H}}
\nc{\sfI}{{\mathsf I}}
\nc{\sfJ}{{\mathsf J}}
\nc{\sfK}{{\mathsf K}}
\nc{\sfL}{{\mathsf L}}
\nc{\sfM}{{\mathsf M}}
\nc{\sfN}{{\mathsf N}}
\nc{\sfO}{{\mathsf O}}
\nc{\sfP}{{\mathsf P}}
\nc{\sfQ}{{\mathsf Q}}
\nc{\sfR}{{\mathsf R}}
\nc{\sfS}{{\mathsf S}}
\nc{\sfT}{{\mathsf T}}
\nc{\sfU}{{\mathsf U}}
\nc{\sfV}{{\mathsf V}}
\nc{\sfX}{{\mathsf X}}
\nc{\sfY}{{\mathsf Y}}
\nc{\sfZ}{{\mathsf Z}}
\nc{\sfW}{{\mathsf W}}

\nc{\idop}{{\mathds{1}}}
\nc{\ve}{{\varepsilon}}

\nc{\RR}{{{\mathbb R}}}
\nc{\CC}{{{\mathbb C}}}
\nc{\FF}{{{\mathbb F}}}
\nc{\NN}{{{\mathbb N}}}
\nc{\MM}{{{\mathbb M}}}
\nc{\ZZ}{{{\mathbb Z}}}
\nc{\PP}{{{\mathbb P}}}
\nc{\QQ}{{{\mathbb Q}}}
\nc{\UU}{{{\mathbb U}}}
\nc{\EE}{{{\mathbb E}}}
\nc{\id}{{\operatorname{id}}}

\nc{\Choi}{Choi-Jamio\l{}kowski }
\nc{\CPTP}{\text{\rm CPTP}}
\nc{\tB}{\widetilde{B}}
\nc{\tA}{\widetilde{A}}
\nc{\tE}{\widetilde{E}}
\nc{\tS}{\widetilde{S}}
\nc{\bB}{\Bar{B}}
\nc{\bA}{\Bar{A}}

\allowdisplaybreaks

\newcommand{\cmark}{\ding{51}}%
\newcommand{\xmark}{\ding{55}}%
\nc{\pT}{{{\mathsf{T}}}}
\nc{\sym}{{\operatorname{sym}}}
\nc{\asym}{{\operatorname{asy}}}

\nc{\SEQ}{\texttt{SEQ}}
\nc{\PAR}{\texttt{PAR}}
\nc{\Comb}{\textsf{Comb}}
\nc{\dia}{\Diamond}
\nc{\bdia}{\blacklozenge}

\definecolor{refcolor}{rgb}{0.067,0.5,0.65}
\definecolor{urlcolor}{rgb}{0.1,0,0.9}
\usepackage[
    breaklinks,
    pdftex,
    colorlinks=true,
    linkcolor=blue,
    citecolor=blue,
    urlcolor=blue
]{hyperref}
\usepackage{cleveref}
\begin{document}

\title{Simulation of Adjoints and Petz Recovery Maps for Unknown Quantum Channels}

\author{Chengkai Zhu}
\thanks{These two authors contributed equally.}
\author{Ziao Tang}
\thanks{These two authors contributed equally.}

\author{Guocheng Zhen}
\affiliation{Thrust of Artificial Intelligence, Information Hub, The Hong Kong University of Science and Technology (Guangzhou), Guangzhou 511453, China}

\author{Yinan Li}
\email{yinan.li@whu.edu.cn}
\affiliation{School of Artificial Intelligence, Wuhan University, Hubei 430072, China}
\affiliation{Hubei Center for Applied Mathematics, Hubei 430072, China}
\affiliation{Hubei Key Laboratory of Computational Science, Hubei 430072, China}
\affiliation{Wuhan Institute of Quantum Technology, Hubei 430206, China}

\author{Ge Bai}
\email{gebai@hkust-gz.edu.cn}
\author{Xin Wang}
\email{felixxinwang@hkust-gz.edu.cn}
\affiliation{Thrust of Artificial Intelligence, Information Hub, The Hong Kong University of Science and Technology (Guangzhou), Guangzhou 511453, China}

\begin{abstract}
Transformations of quantum channels, such as the transpose, complex conjugate, and adjoint, are fundamental to quantum information theory. Given access to an unknown channel, a central problem is whether these transformations can be implemented physically with quantum supermaps. While such supermaps are known for unitary operations, the situation for general quantum channels is fundamentally different. 
In this work, we establish a strict hierarchy of physical realizability for the transposition, complex conjugation, and adjoint transformation of an unknown quantum channel. We present a probabilistic protocol that exactly implements the transpose with a single query. In contrast, we prove no-go theorems showing that neither the complex conjugate nor the adjoint can be implemented by any completely positive supermap, even probabilistically. We then overcome this impossibility by designing a virtual protocol for the complex conjugate based on quasi-probability decomposition, and show its optimality in terms of the diamond norm. As a key application, we propose a protocol to estimate the expectation values resulting from the Petz recovery map of an unknown channel, achieving an improved query complexity compared to existing methods.
\end{abstract}
\maketitle

\section{Introduction}

The theory of quantum information has largely centered on the control and manipulation of quantum states. A broader paradigm is to manipulate quantum processes themselves, the dynamical maps that govern state evolution. This \textit{second-order} control, formalized by quantum supermaps~\cite{Chiribella_2008,Chiribella_2008a}, asks which transformations of unknown physical processes are fundamentally possible. A canonical testbed for this theory comprises the following transformations of an unknown channel: the adjoint, complex conjugate, and transpose. The adjoint is particularly indispensable. By the duality relation
$\tr[O\,\cN(\rho)] = \tr[\cN^\dagger(O)\,\rho]$, it governs Heisenberg-picture evolution~\cite{Angrisani2025,Schuster2025} and, operationally, corresponds to the ``backward time evolution'' needed to measure out-of-time-ordered correlators (OTOCs), a key diagnostic of quantum scrambling in open systems~\cite{Syzranov2018,Zanardi2021,Schuster2023,Bose_2024,Xu2024}. The adjoint also underpins the Petz recovery map~\cite{Petz1986,PETZ1988,1993Qeai,Barnum2002,Hayden_2004}, and thus the characterization of reversibility for quantum dynamics, which is central to near-optimal error-correction protocols~\cite{Ng2010,Tyson2010,Biswas_2024} and to thermodynamic limits of quantum information processing~\cite{JORDAN20102075}. The special case of obtaining the adjoint (namely the inverse) of unitary channels is essential for quantum algorithms such as quantum phase processing~\cite{Wang2023} and quantum singular value transformation~\cite{gilyen2019quantum}.
Besides applications in quantum information and algorithms, such transformations are essential for grounding the mathematical theory of time-reversal in open quantum systems \cite{Chiribella_2021,Buscemi2021,aw2021fluctuation}. 
This leads to a fundamental question: given a black-box device implementing an unknown channel $\cN_{A\to B}$, can an experimentalist implement its complex conjugate, transpose, or adjoint?

For unitary evolutions, this question is largely understood. In particular, mapping from $U$ to $U^\dag$ corresponds to time reversal. Chen et al.~\cite{chen2025quantumalgorithmreversingunknown} gave an exact circuit-level protocol that, with $O(d^2)$ queries to the unknown unitary, deterministically and exactly reverses an unknown $ d$-dimensional unitary, and is proven to be optimal in query complexity~\cite{odake2024}. There have also been extensive studies on probabilistic and approximate protocols~\cite{Chiribella_2016,Quintino_2019,Quintino2019,Michal2019}.

However, for general open-system dynamics described by quantum channels, the situation changes qualitatively. For a generic quantum channel $\cN_{A\to B}$, the conjugate $\cN^*_{A\to B}$ remains a quantum channel, i.e., a completely positive (CP) and trace-preserving (TP) map, while the transpose $\cN^\pT$ and the adjoint $\cN^\dag$ are generally not TP and thus are not channels. This raises a subtle and fundamental problem: can these non-TP maps nonetheless be realized operationally, i.e., perhaps probabilistically, or via a virtual (quasi-probabilistic) scheme?

In this paper, we resolve this question by establishing a strict hierarchy for the physical realizability of the adjoint, complex conjugate, and transpose of a black-box channel. We begin by demonstrating that while the channel transpose $\cN^\pT$ is physically realizable probabilistically via a post-selected teleportation protocol, the complex conjugate $\cN^*$ and adjoint $\cN^\dag$ are fundamentally different. We prove a no-go theorem ruling out their implementation by any completely positive supermap, even probabilistically. To overcome this physical barrier, we introduce a virtual protocol based on quasi-probability decompositions that enables universal estimation of these unrealizable maps. Together, our results (summarized in Table~\ref{tab:realization}) provide a complete picture for the physical implementation of universal complex conjugation, transpose, and adjoint of quantum channels, within the broader context of second-order quantum transformations. 

As an application, our method for simulating the adjoint map of an unknown quantum channel enables the implementation of the Petz recovery map using only black-box access, building upon the previous subroutines~\cite{gilyen2022quantum}. Specifically, we provide a quasi-sampling protocol to estimate the expectation value $\tr[O_A\cP_{\sigma,\cN}(\omega_B)]$ for any observable $O_A$, input state $\omega_B$ and any unknown quantum channel $\cN_{A\to B}$ given a prior state $\sigma$, where $\cP_{\sigma,\cN}$ denotes the Petz map. Our estimator achieves error at most $\varepsilon$ with probability at least $1-\delta$ using $O\left(\frac{d_A^4d_B^3}{ \varepsilon^2}\log\left(\frac{1}{\delta}\right)\right)$ queries to the unknown channel when the channel is unital, where $d_A$ and $d_B$ are the input and output dimensions of the channel, respectively. Notably, for expectation‑value estimation, this improves upon the deterministic approximation of the Petz map in Ref.~\cite{utsumi2025}, whose overall query complexity scales as $O\left(\frac{d_A^{5.5}d_B^{2.5}}{\varepsilon^4\lambda_{\min}^{3/2} } \log\left(\frac{1}{\delta}\right)\min{\left\{ \frac{1}{\tau^2_{\min}}, \frac{d_A^8 d_B^4}{\varepsilon^4 \lambda_{\min}^2 }\right\}}\right)$, where $\lambda_{\min}$ and $\tau_{\min}$ are channel related to parameters that are also dependent on $d_A,d_B$.

\begin{table}[t!]\label{tab:realization}
\centering
\renewcommand{\arraystretch}{1.2}
\begin{tabular}{>{\bfseries} l @{\hspace{10pt}} c @{\hspace{12pt}} c @{\hspace{12pt}} c}
\toprule
& \textbf{$\cN^*$} & \textbf{$\cN^{\pT}$} & \textbf{$\cN^\dagger$} \\
\midrule
\textbf{Quantum comb} & \xmark~(Thm.~\ref{thm:channel_conj_CPnogo}) & \xmark~(TP) & \xmark~(TP) \\
\textbf{Prob. comb} & \xmark~(Thm.~\ref{thm:channel_conj_CPnogo}) & \cmark~(FIG.~\ref{fig:channel_tele}) & \xmark~(Thm.~\ref{thm:channel_conj_CPnogo}) \\
\textbf{Virtual comb} & \cmark~(Thm.~\ref{thm:conj_HPTP}) & \xmark~(TP) & \xmark~(TP) \\
\bottomrule
\end{tabular}
\caption{Feasibility of universally realizing the complex conjugate $\cN^*$, the transpose $\cN^{\pT}$, and the adjoint $\cN^{\dag}$ of an unknown quantum channel $\cN$. Each row specifies the allowed class of higher-order transformations. "\xmark(TP)" indicates a no-go because the target map is generally not trace-preserving. "\cmark~(FIG.~\ref{fig:channel_tele})": the postselected teleportation for $\cN^{\pT}$. "\xmark~(Thm.~\ref{thm:channel_conj_CPnogo})": no-go for CP realization of $\cN^*$, and hence $\cN^{\dag}$. "\cmark~(Thm.~\ref{thm:conj_HPTP})": virtual-comb realization of $\cN^*$.
}
\end{table}

\section{Notation}
A quantum system $A$ is associated with a $d_A$-dimensional Hilbert space $\cH_A$. For a bipartite system $A'A$ where $\cH_{A'}\cong \cH_A$, we denote by $\Phi_{A'A}=\sum_{i,j=0}^{d-1}\ketbra{ii}{jj}_{A'A}$ the maximally entangled state on $A'A$, and by $F_{A'A}$ the SWAP operator, i.e., $F_{A'A}\ket{ij} = \ket{ji}$. We denote by $\mathscr{L}_{A}$ the set of linear operators on $\cH_A$, by $\mathscr{H}_A$ the set of Hermitian operators, by $\mathscr{P}_A$ the set of positive semidefinite operators, and by $\mathscr{D}_A$ the set of quantum states.

A linear map $\cN_{A\to B}$ from $\mathscr{L}_{A}$ to $\mathscr{L}_B$ is called completely positive (CP) if $(\cI_R\ox \cN_{A\to B})(X) \in \mathscr{P}_{RA}$ for any $X \in \mathscr{P}_{RA}$ and any ancillary system $R$, is called Hermitian-preserving (HP) if $\cN_{A\to B}(X) \in \mathscr{H}_A$ for any $X \in \mathscr{H}_A$, is called trace-preserving (TP) if $\tr [\cN_{A\to B}(X)] = \tr X$ for any $X \in\mathscr{L}_{A}$, and is called trace-nonincreasing (TN) if $\tr [\cN_{A\to B}(X)] \leq \tr X$ for any $X\in\mathscr{P}_A$. A quantum channel $\cN_{A\to B}$ is a CPTP linear map from $\mathscr{L}_{A}$ to $\mathscr{L}_{B}$. 
We denote by $\CPTP(A,B)$ the set of all quantum channels from $A$ to $B$. 
Throughout the paper, we will use calligraphic letters to denote linear maps and sans-serif versions to denote their corresponding \Choi operators. For a linear map $\cN_{A\to B}$, its Choi operator is defined by $\sfN_{A'B} = (\cI_{A'}\ox\cN_{A\to B})(\Phi_{A'A})$. For a channel with Kraus representation $\cN_{A\to B}(\cdot) = \sum_j K_j^{}(\cdot)K_j^\dag$, its \textit{complex conjugate}, \textit{transpose}, and \textit{adjoint} are respectively defined by 
\begin{equation}
\begin{aligned}
& \cN^*_{A\to B}(\cdot) = \sum_j K^*_j(\cdot) K_j^{\pT},\\
& \cN^{\pT}_{B\to A}(\cdot) = \sum_j K_j^{\pT}(\cdot)K_j^{*},\\
& \cN^{\dagger}_{B\to A}(\cdot) = \sum_j K_j^{\dag}(\cdot)K_j^{},
\end{aligned}
\end{equation}
where the complex conjugate and the transpose are defined with respect to a predetermined orthonormal basis.
Following this definition, it can be directly checked that the Choi operators of $\cN_{A\to B}^*, \cN_{B\to A}^{\pT}$ and $\cN_{B\to A}^\dag$ are given by $\sfN_{A'B}^{\pT},F_{A'B}\sfN_{A'B}F_{A'B}^\dag$, and $F_{A'B}\sfN_{A'B}^{\pT}F_{A'B}^\dag$, respectively.

A linear supermap $\cC_{(A\to B)\to (A'\to B')}$ transforms linear maps from $\mathscr{L}_{A}$ to $\mathscr{L}_{B}$ to linear maps from $\mathscr{L}_{A'}$ to $\mathscr{L}_{B'}$. A linear supermap $\cC_{(A\to B)\to (A'\to B')}$ is called HP if $\cC_{(A\to B)\rightarrow (A'\to B')}(\cN_{A\to B})$ is HP for any HP input $\cN_{A\to B}$, and is called CP if $(\mathbb{I}_R \ox \cC_{(A\to B)\rightarrow (A'\to B')})(\cN_{R_0A\to R_1B})$ is CP for any CP input $\cN_{R_0A\to R_1B}$, where the reference systems $R_0,R_1$ are arbitrary and $\mathbb{I}_R$ denotes the identity supermap for linear maps from $\mathscr{L}_{R_0}$ to $\mathscr{L}_{R_1}$. 
More generally, an $n$-slot supermap $\cC$ maps quantum channels $\cN^{(i)}_{I_i\to O_i}$ for $i=1,2,\dots,n$ to an output quantum channel $\cN_{P\to F}$. Since any such supermap is equivalent to a linear map $\cN^{\cC}_{PO_1\cdots O_n\to I_1\cdots I_nF}$, we call a supermap HP (CP, TP) if its corresponding linear map is HP (CP, TP)~\cite{Chiribella_2008a,Zhu_2024,xiao2025}. Notably, an $n$-slot CPTP supermap is known as an $n$-slot \textit{quantum comb}~\cite{Chiribella_2008}.

\section{Results}

\subsection{Probabilistic Transposition and No-Go for Complex Conjugation}

We first observe that the transpose $\cN^\pT_{B\to A}$ of an unknown quantum channel can be probabilistically realized via postselected teleportation~\cite{Lloyd2011}, also known as gate teleportation in the unitary case~\cite{Gottesman_1999,Quintino_2019}, as depicted in FIG.~\ref{fig:channel_tele}. The procedure is: i). Prepare a maximally entangled state $\Phi_{A'A}$ and the input state $\rho_{B'}$, ii). Apply the unknown channel $\cN_{A\to B}$ on system $A$, iii). Perform a measurement $\left(\Phi_{B'B} , \idop-\Phi_{B'B}\right)$ on systems $B'B$. An outcome corresponding to the projector onto $\Phi_{B'B}$ heralds the successful preparation of the state $\cN^\pT(\rho)$ on system $A'$.

We can calculate that for any input state $\rho_{B'}$, the output unnormalized state on system $A'$ reads
\begin{equation*}
\begin{aligned}
    &\frac{1}{d_A}\tr_{BB'}\Big[(\Phi_{BB'}\ox \idop_{A'}) (\sfN_{AB}\ox \rho_{B'})\Big]\\
    =&\frac{1}{d_Ad_B}\tr_{B}\big[\sfN_{AB}(\rho_{B'}^{\pT}\ox\idop_{A})\big] = \frac{1}{d_Ad_B}\cN_{B'\to A'}^{\pT}(\rho_{B'}).
\end{aligned}
\end{equation*}
Thus, the output state is $\cN_{B'\to A'}^{\pT}(\rho_{B'})$, and the success probability of this protocol is given by $p_{\text{suc}} =\tr\big[\sfN_{AB}(\rho_{B'}^{\pT}\ox\idop_{A})\big]/(d_Ad_B)$. When $\cN_{A\to B}$ is a unital channel, we have that $p_{\text{suc}} =1/d^2$ where $d_A=d_B = d$.

While the transpose admits a probabilistic implementation, the situation changes strictly for the complex conjugate. We show that a universal realization of $\cN^*$ is impossible within the domain of physical supermaps, even with any finite copy supplied or allowing for probabilistic success.

\begin{theorem}[No-go for complex conjugation]\label{thm:channel_conj_CPnogo}
For any finite integer $n \ge 1$, there exists no $n$-slot completely positive supermap $\cC$ such that $\cC(\cN^{\ox n}) = \cN^*,~\forall \cN\in\CPTP(A,B)$.
\end{theorem}

The proof is by contradiction. It was shown that for any integer $d,n \ge 1$, there exists no $n$-to-$1$ completely positive map $\cN_{A^n\to A'}$ such that $\cN(\rho^{\ox n}) = \rho^{\pT}$ for all $\rho\in\mathscr{D}_A$~\cite{Yang_2014,Dong_2019}. Then, suppose there is a CP supermap that can universally realize the complex conjugate with $n$ copies of the input channel for some $n\geq 1$, i.e., $\cC(\cN^{\ox n}) = \cN^*,~\forall \cN\in\CPTP(A,B)$. Consider a state preparation channel $\cR^{\rho}_{A\to B}$ such that $\cR^{\rho}_{A\to B}(X) = \tr (X) \rho,\forall X$. Note that the complex conjugate of $\cR^{\rho}_{A\to B}$ is also a state preparation channel $\cR^{\rho^{\pT}}_{A\to B}(X) = \tr (X) \rho^{\pT},\forall X$. By assumption, we have that $\cC((\cR^{\rho}_{A\to B})^{\ox n}) = \cR^{\rho^{\pT}}_{A\to B}$. Therefore, we have a protocol that can prepare $\rho^{\pT}$ for any quantum state $\rho$, a contradiction. 

A direct consequence of Theorem~\ref{thm:channel_conj_CPnogo} is that the channel adjoint $\cN^\dagger$ is also physically unrealizable from black-box accesses of $\cN$. If a universal CP supermap for the adjoint existed, one could compose it with the probabilistic protocol for the channel transpose to realize the conjugate map ($\cN^* = (\cN^\pT)^\dagger$) probabilistically, contradicting Theorem~\ref{thm:channel_conj_CPnogo}. Hence, we have seen a fundamental gap between realizing the adjoint map for unknown unitaries and unknown channels.

\begin{figure}
    \centering
    \includegraphics[width=.85\linewidth]{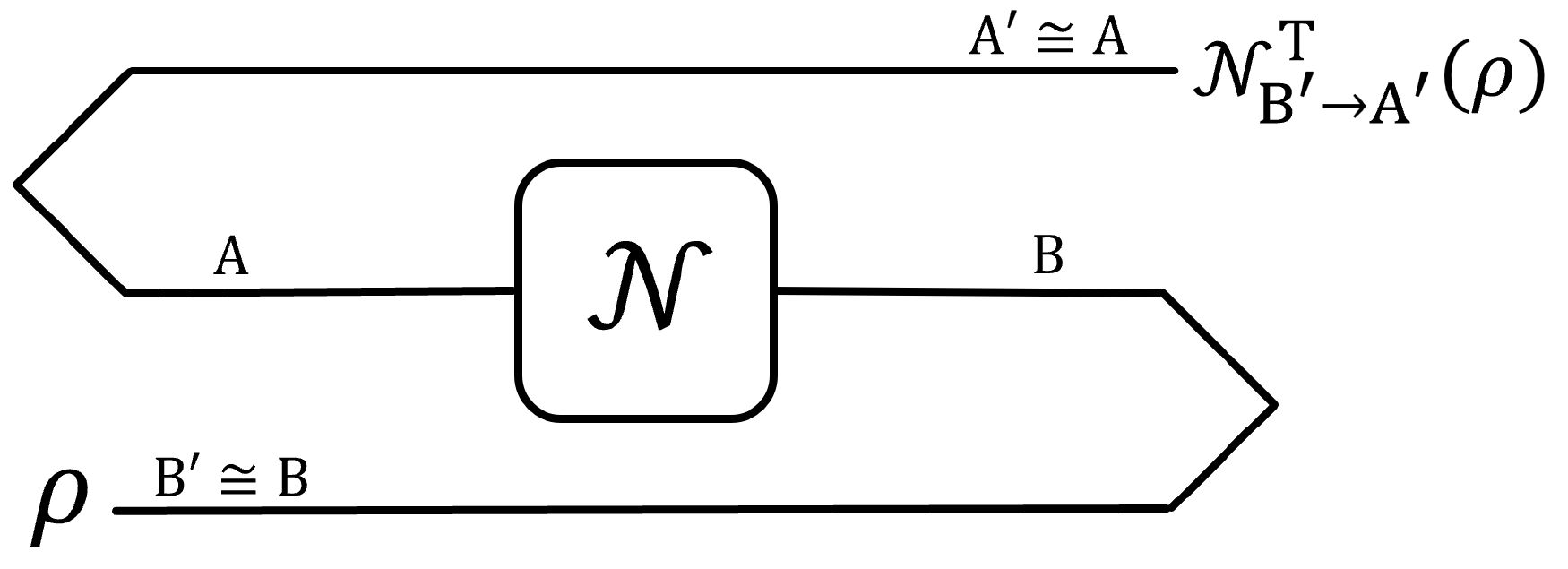}
    \caption{Probabilistic implementation of the transpose of an unknown channel $\cN_{A\to B}$. A maximally entangled state "$\langle$" and an input state $\rho$ are prepared. The black-box channel $\cN_{A\to B}$ is applied and a Bell measurement $\{\Phi_{B'B}, \idop-\Phi_{B'B}\}$ is then performed. The output state on system $A'$ is $\cN_{B\to A}^{\pT}(\rho)$ when the measurement outcome is $\Phi_{B'B}$.}
    \label{fig:channel_tele}
\end{figure}

\subsection{Simulation of Channel Adjoint via Werner-Holevo Channels}

The obstruction of Theorem~\ref{thm:channel_conj_CPnogo} can be traced to the impossibility of universally implementing matrix transposition on arbitrary quantum states. That is, if we can realize the transpose map for unknown Choi states, we can implement the corresponding channel via postselected teleportation~\cite{Lloyd2011}. Nevertheless, there exist canonical CPTP channels that contain the transpose as a component, namely the Werner–Holevo channels~\cite{Werner2002}
\begin{equation}
\cW_d^{\pm}(X) = \frac{1}{d\pm 1} \big[\tr (X) \idop_d \pm X^{\pT}\big].
\end{equation} 
These channels are fully physical yet explicitly mix the desired transpose term $X^\pT$ with an isotropic noise term. By taking a linear combination of $\cW^+_d$ and $\cW^-_d$, one can effectively eliminate the identity term $\tr(X)\idop_d$ and isolate the transpose map. The price of ``cancelling noise'' is that the required combination is generally not a convex mixture, and hence cannot define a CP supermap. To bypass this, we adopt the framework of \textit{virtual combs}~\cite{Zhu_2024}. By relaxing the convexity requirement, we can construct quasi-probabilistic mixtures of pre- and post-processing Werner-Holevo channels that isolate the desired non-physical mappings.

Specifically, an $n$-slot virtual comb admits a quasi-probability decomposition
\begin{equation}
    \widetilde{\cC}=\sum_{i} \alpha_i \cC_i,~~\sum_i \alpha_i = 1,
\end{equation}
where each $\cC_i$ is a $n$-slot quantum comb and $\alpha_i\in\RR$. To quantify how much a given virtual comb deviates from physical combs, we introduce its \textit{base norm} with respect to the set of quantum combs as follows.
\begin{equation*}\label{Eq:basenorm}
\|\widetilde{\cC}\|_{\bdia} \!\coloneqq \!\min_{\alpha_{\pm}\geq 0} \Big\{\alpha_+ + \alpha_-\!: \widetilde{\cC}=\alpha_+ \cC_+ - \alpha_- \cC_-,\cC_{\pm} \in\mathsf{Comb}_n\Big\},
\end{equation*}
where $\mathsf{Comb}_n$ denotes the set of all $n$-slot quantum combs. Similar quantities have already been studied for HP linear maps~\cite{Jiang_2021,regula2021operational}. We show that this base norm is nothing but the diamond (completely bounded trace) norm for combs, defined operationally by (cf.~\cite[Eq.~(19)]{Chiribella_2008m} and~\cite[Definition 1]{gutoski2012measure})
\begin{equation}
    \|\widetilde{\cC}\|_{\dia} \!\coloneqq \max_{\Xi^{(n)}}\left\|\left(\idop\ox \sqrt{\Xi^{(n)}}\right)\widetilde{\sfC}\left(\idop\ox \sqrt{\Xi^{(n)}}\right)\right\|_1,
\end{equation}
where $\widetilde{\sfC}$ is the Choi operator of $\widetilde{\cC}$ and $\Xi^{(n)}$ defines an $n$-slot \textit{quantum tester} $\{\sfT_i\}_i$ for $\widetilde{\cC}$ by $\sfT_i\geq 0$,
\begin{equation}
\begin{aligned}
    &\sum_i \sfT_i = \idop_{F} \ox \Xi^{(n)},\\
    &\tr_{O_k}[\Xi^{(k)}] = \idop_{I_k}\ox \Xi^{(k-1)},~\forall 2\leq k\leq n,\\
    &\tr \Xi^{(1)} = 1.
\end{aligned}
\end{equation}
The diamond norm admits an operational interpretation through comb discrimination~\cite{Chiribella_2008m,gutoski2012measure}, and generalizes the diamond norm of an HP linear map $\cN$, i.e., $\|\cN\|_{\dia} = \max_{\rho}\left\|\left(\idop\ox \sqrt{\rho}\right)\sfN\left(\idop\ox \sqrt{\rho}\right)\right\|_1$.
\begin{lemma}\label{lemma:basenorm=diamondnorm}
    For any $n$-slot virtual comb $\widetilde{\cC}$, it holds that $\|\widetilde{\cC}\|_{\bdia} = \|\widetilde{\cC}\|_{\dia}$.
\end{lemma}
The proof is detailed in Appendix~\ref{app:proof}. Now, we present a virtual realization of the universal complex conjugate of quantum channels, utilizing the Werner-Holevo channels as pre- and post-processing channels.

\begin{theorem}[Virtual complex conjugation]\label{thm:conj_HPTP}
There exists a 1-slot virtual comb $\widetilde{\cC}$ that universally implements the complex conjugate, i.e., $\widetilde{\cC}(\cN) = \cN^*,\forall \cN \in \CPTP(A,B)$. 
\end{theorem}

\begin{proof}
The proof is constructive. 
Let us denote by
\begin{equation}
P_{A'A}^s = \frac{1}{2}(\idop_{AA'} + F_{AA'}), \quad P_{A'A}^a = \frac{1}{2}(\idop_{A'A} -  F_{AA'})
\end{equation}
the projection onto the symmetric and antisymmetric subspaces on $A'A$, respectively. The Choi operators of the Werner-Holevo channels $\cW_{d_A}^{+}$ and $\cW_{d_A}^{-}$ are given by
\begin{equation}
\sfW_{d_A}^{+} = \frac{2}{d+1} P^s_{A'A},\quad \sfW_{d_A}^{-} = \frac{2}{d-1} P^a_{A'A},
\end{equation}
respectively. Now, we will construct a desired virtual comb $\widetilde{\cC}_{(A\to B)\to (A'\to B')}$ to realize the complex conjugation of any unknown channel. Consider the following three operators 
\begin{equation}\label{Eq:Choi_C123}
\begin{aligned}
&\sfC^{(1)}_{A'ABB'}\coloneqq \sfW_{d_A}^{+} \ox \sfW_{d_B}^{+},\\
&\sfC^{(2)}_{A'ABB'}\coloneqq \sfW_{d_A}^{-} \ox \sfW_{d_B}^{-},\\
&\sfC^{(3)}_{A'ABB'}\coloneqq \sfW_{d_A}^{+} \ox \sfW_{d_B}^{-},
\end{aligned}
\end{equation}
and a supermap $\widetilde{\cC}_{(A\to B)\to (A'\to B')}$ with a Choi operator
\begin{equation}\label{Eq:virtualC}
\begin{aligned}
\widetilde{\sfC}_{A'ABB'} &\coloneqq \frac{d_B+1}{2}\sfC^{(1)}_{A'ABB'} + \frac{(d_A-1)(d_B-1)}{2}\sfC^{(2)}_{A'ABB'}\\
&\qquad - \frac{d_A(d_B-1)}{2}\sfC^{(3)}_{A'ABB'}\\
&= \frac{2}{d_A+1} P^s_{A'A}\ox P^s_{BB'} + 2 P^a_{A'A}\ox P^a_{BB'}\\
&\qquad - \frac{2d_B}{d_A+1} P^s_{A'A}\ox P^a_{BB'} \\
&= F_{A'A} \ox F_{BB'} + \frac{\idop_{A'ABB'}}{d_A+1} - \frac{d_A(F_{A'A} \ox \idop_{BB'})}{d_A+1}.
\end{aligned}
\end{equation}
This constitutes a valid virtual comb, as verified by
\begin{equation*}
\begin{aligned}
\tr_{B'}\widetilde{\sfC}_{A'ABB'} &= \left(1- \frac{d_B^2}{d_A+1}\right)F_{A'A} \ox \idop_{B} + \frac{d_B}{d_A+1}\idop\\
&= \tr_{BB'}\widetilde{\sfC}_{A'ABB'} \ox \frac{\idop_{B}}{d_B},\\
\tr_{AB'}\widetilde{\sfC}_{A'ABB'} &= \frac{d_A+1+d_Ad_{B'}-d_Ad_B}{d_A+1}\idop_{A'B}  = \idop_{A'B}. 
\end{aligned}
\end{equation*}
Applying $\widetilde{\cC}_{(A\to B)\to (A'\to B')}$ to any quantum channel $\cN_{A\to B}$, the Choi operator of the output map can be calculated as
\begin{equation}
\begin{aligned}
&\tr_{AB}\left[\widetilde{\sfC}_{A'ABB'}\left(\sfN_{AB}^{\pT}\ox \idop_{A'B'}\right)\right]\\
=&\; \sfN_{AB}^{\pT} + \frac{d_A}{d_A+1} \idop_{A'B'} - \frac{d_A}{d_A+1}\idop_{A'B'} = \sfN_{AB}^{\pT}.
\end{aligned}
\end{equation}
\end{proof}

\begin{figure}[t]
    \centering
    \includegraphics[width=1\linewidth]{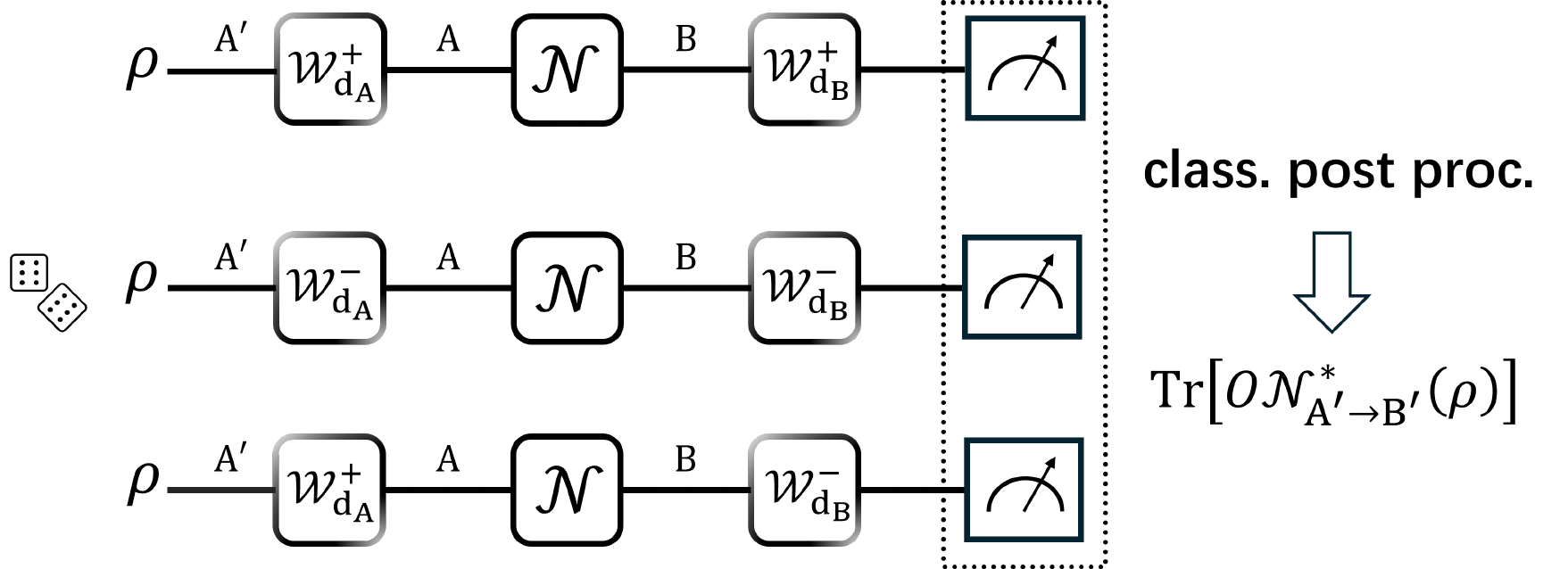}
    \caption{Quasi-probabilistic simulation of the complex conjugate $\cN^*$ for an unknown channel $\cN$. The protocol brackets the channel with pre- and post-processing Werner-Holevo channels $\cW^{\pm}$, selected according to a specific quasiprobability distribution. The target expectation value $\tr[O\cN^*(\rho)]$ is reconstructed through classical post-processing of the measurement outcomes from multiple sampling rounds.}
    \label{fig:conjugate}
\end{figure}

The virtual comb constructed in Theorem~\ref{thm:conj_HPTP} translates directly into an operational sampling protocol. Since the virtual comb is a linear combination of three physical processes, with pre- and post-processing channels being 
\begin{equation*}
\{\cW_{d_A}^{+}(\cdot),\cW_{d_B}^{+}(\cdot)\},\{\cW_{d_A}^{-}(\cdot),\cW_{d_B}^{-}(\cdot)\},\{\cW_{d_A}^{+}(\cdot),\cW_{d_B}^{-}(\cdot)\},
\end{equation*}
expectation values of the conjugate map $\tr[O\cN^*(\rho)]$ can be estimated via a Monte Carlo procedure, as illustrated in FIG.~\ref{fig:conjugate} to estimate $\tr[O \cN_{A\to B}^*(\rho)]$ for any observable $O$ and state $\rho$. For $O$ diagonal in computational basis (general $O$ can be reduced to a diagonal one by inserting a unitary before measurement), i.e., $O=\sum_{x\in\{0,1\}^n} A(x)\ketbra{x}{x},A(x)\in[-1,1]$, the detailed protocol is as follows. Let $\gamma=d_Ad_B-d_A+1$ and
\begin{equation}\label{eq:sample_probs}
    p_1 = \frac{d_B+1}{2\gamma},p_2 = \frac{(d_A-1)(d_B-1)}{2\gamma},p_3 = \frac{d_A(d_B-1)}{2\gamma}.
\end{equation}
We sample $(x,y)$ from $\{(+,+),(-,-),(+,-)\}$ with probability $p_1,p_2,p_3$, respectively. 
In the $m$-th out of $M$ sampling rounds, apply $\cW_{d_B}^{y} \circ \cN_{A\to B} \circ \cW_{d_A}^{x}(\cdot)$ on the input state, measure the output state in the computational basis and denote the measurement result as $s^{(m)}$. The $m$-th sample results in a random variable 
\begin{equation}
X^{(m)}  = 
\begin{cases}
\gamma A(s^{(m)}) , &\text{if } x= y, \\
-\gamma A(s^{(m)}), &\text{if } x\neq y,
\end{cases}
\end{equation}
where $A(s^{(m)})$ is the measurement value. The empirical mean of $X^{(m)}$ is given by $\xi = 1/M\sum_{m=1}^M X^{(m)}$, which is an unbiased estimator of $\tr\left[O\cN_{A\to B}^*(\rho) \right]$ (cf.~\cite[Lemma 16]{Jiang_2021}). By Hoeffding's inequality, the number of samples $M$ required to achieve precision $\varepsilon$ with probability $1-\delta$ scales as $M = O(d_A^2 d_B^2 \log(1/\delta)/\varepsilon^2)$, where the scaling factor $O(d_A^2 d_B^2)$ is related to the base norm for a virtual comb that bounds the difference between measurement outcomes.

From above, we can see that the variance and thus the sample complexity of the resulting estimator is controlled by the base norm of the virtual comb. We further show that the construction in Eq.~\eqref{Eq:virtualC} is the optimal realization with respect to the base norm. 
\begin{theorem}\label{thm:min_basenorm}
The minimum base norm of a 1-slot virtual comb that satisfies $\widetilde{\cC}(\cN)=\cN^*,~\forall \cN\in\CPTP(A,B)$ is $d_Ad_B-d_A+1$.
\end{theorem}
The proof is detailed in Appendix~\ref{app:proof}. Naturally, Theorem~\ref{thm:conj_HPTP} provides a path to universally implement the adjoint map virtually. Since $\cN^\dag = (\cN^*)^{\pT}$, we can combine the virtual protocol for the complex conjugate with the probabilistic protocol for the transpose.

\subsection{Petz map of unknown quantum channel}
A central motivation for universally realizing the adjoint map is the Petz recovery map~\cite{Petz1986,PETZ1988}, which is a canonical reversal operation associated with a channel $\cN_{A\to B}$ and a reference state $\sigma_A$, given by \[\cP_{\sigma, \cN}(\cdot) \coloneqq \sigma^{1/2}\cN^\dag\Big(\cN(\sigma)^{-1/2}(\cdot)\cN(\sigma)^{-1/2}\Big)\sigma^{1/2}.\]
We can decompose $\cP_{\sigma,\cN}$ into the following CP maps:
\begin{itemize}
    \item[i).] $\cN(\sigma)^{-1/2}(\cdot)\cN(\sigma)^{-1/2}$, 
    \item[ii).] $\cN^\dag(\cdot)$,
    \item[iii).] $\sigma^{1/2}(\cdot)\sigma^{1/2}$.
\end{itemize}
Overall, $\cP_{\sigma,\cN}$ is a CPTP map when $\cN(\sigma)$ is full-rank. It provides a near-optimal reversal of a quantum channel, serving as a quantum analog of Bayes' rule~\cite{Bai_2025, liu2025unifyingquantumsmoothingtheories}, which makes it a fundamental tool for quantum information recovery when perfect restoration is impossible~\cite{Wilde_2015,Junge_2018}.

To physically realize the Petz map, Ref.~\cite{gilyen2022quantum} proposed an algorithmic approach to realize the three steps above separately, which is built on the block-encoding technique~\cite{low2019hamiltonian} and QSVT~\cite{gilyen2019quantum,gilyen2019thesis}. Such an implementation requires access to copies of the Stinespring isometry $V$ (and $V^\dag$) of the quantum channel $\cN$, which could be challenging and a strong assumption in practice. Crucially, by querying only the \textit{unknown} black-box channel, we develop a protocol to estimate expectation values $\tr[O_A \cP_{\sigma,\cN}(\omega_B)]$, for arbitrary input $\omega_{B}$ and observable $O_A$. As shown in FIG.~\ref{fig:petzmap}, we outline the protocol for estimating $\tr\left[O_A\cP _{\sigma,\cN}(\omega_B)\right]$ as follows. In each of the $M = O(d_A^2 d_B^2 \log(1/\delta)/\varepsilon^2)$ sampling rounds:
\begin{itemize} 
    \item[i).] Sample $(x,y)\in\{(+,+),(-,-),(+,-)\}$ with the probabilities $p_1,p_2,p_3$ in Eq.~\eqref{eq:sample_probs} respectively, and prepare the input state $\ketbra{0}{0}_{R'}\ox \Phi_{A'A''} \ox \omega_{B''} \ox \ketbra{0}{0}_{R}$.
    \item[ii).] Apply $\cW_{d_B}^{y} \circ \cN_{A\to B} \circ \cW_{d_A}^{x}(\cdot)$ on system $A'$, the block-encoding $U_{RB''}^{\cN (\sigma)^{-1/2}}$ on system $RB''$, and the block-encoding $U_{R'A''}^{\sigma^{1/2}}$ on system $R'A''$.
    \item[iii).] Perform a projective measurement $\big(\Phi_{B''B'}, \idop-\Phi_{B''B'}\big)$ on systems $B''B'$ and measure in computational basis on system $R'$ and $R$. Take measurement on system $A''$ with post-processing procedure by post-selecting the outcome $\ketbra{0}{0}_{R'} \ox \Phi_{B''B'}\ox \ketbra{0}{0}_{R}$.
\end{itemize}
After applying the same post-processing procedure as in Theorem~\ref{thm:conj_HPTP}  on $A''$ to obtain the sample results $\{\widetilde{X}^{(m)}\}_{m=1}^M$, we arrive at an unbiased estimator $\widetilde{\xi} = \frac{1}{M}\sum_{m=1}^{M}\widetilde{X}^{(m)}$ for $\tr\left[O_A\cP _{\sigma,\cN}(\omega_B)\right]$, which achieves precision $\varepsilon$ with probability at least $1-\delta$. Given that the focus of many quantum information processing tasks is inherently classical statistics~\cite{Aaronson2018,Huang2020,Yuan2024,Zhang2024}, this estimation capability effectively captures the observable properties of the recovered state, with an intimate connection to quantum error mitigation~\cite{Temme2017,Cai2023}. We emphasize that our method generalizes to the estimation of any linear functional of $\cP_{\sigma,\cN}(\omega_B)$, though, without physically preparing the state.

\begin{figure}[t]
    \centering
    \includegraphics[width=1\linewidth]{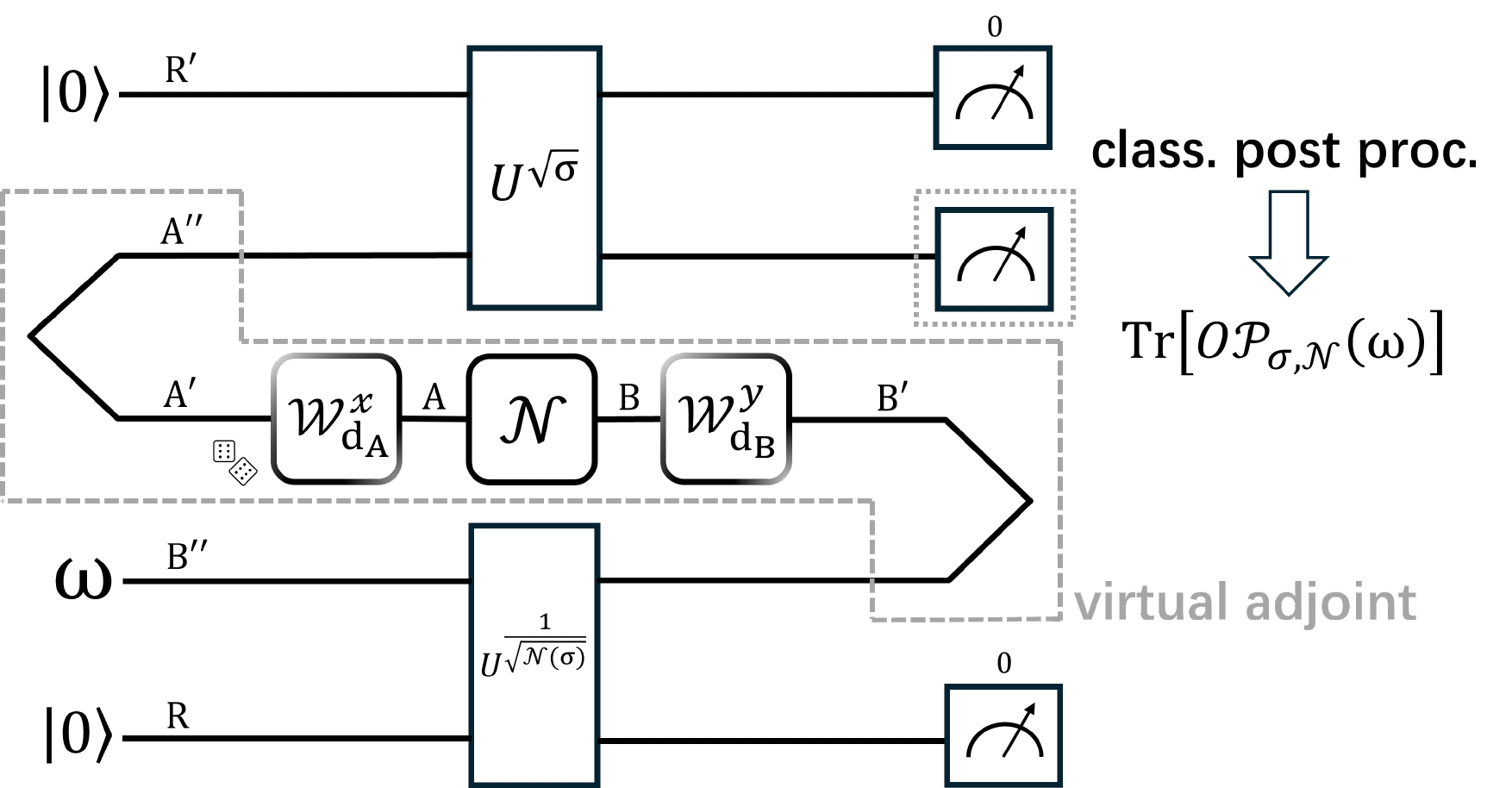}
    \caption{
    Protocol for estimating $\tr[O_A\cP_{\sigma,\cN}(\omega_B)]$ of the Petz recovery map. The scheme integrates the sampling-based simulation of the channel adjoint $\cN^\dagger$ and probabilistic transposition, with block-encodings of $\sigma^{1/2}$ and $\cN(\sigma)^{-1/2}$.}
    \label{fig:petzmap}
\end{figure}

The key conceptual difference from Ref.~\cite{gilyen2022quantum} is that their coherent implementation of $\cN^\dag$ relies on a unitary dilation of $\cN$. In contrast, here $\cN^\dag$ is accessed through our virtual-adjoint procedure derived from Theorem~\ref{thm:conj_HPTP} combined with the probabilistic transpose. Recently, Ref.~\cite{utsumi2025} developed an algorithm to implement the Uhlmann transformation, which can be utilized to realize the Petz map of an unknown channel approximately and deterministically, based on the framework in Ref.~\cite{gilyen2022quantum}. This could be used to achieve our task of estimating $\tr[O_A \cP_{\sigma,\cN}(\omega_B)]$ for unknown channels. We will compare the query complexity of the unknown channel of their method and ours in the following parts.

Specifically, Ref.~\cite{utsumi2025} developed an algorithm to realize the Petz map $\cP_{\sigma,\cN}$ with error $\varepsilon_0$ in the diamond norm distance. This approximation requires $O\left(\frac{d_A^{11/2}d_B^{5/2}}{\varepsilon_0^2\lambda_{\min}^{3/2} } \min{\left\{ \frac{1}{\tau^2_{\min}}, \frac{d_A^8 d_B^4}{\varepsilon_0^4 \lambda_{\min}^2 }\right\}}\right)$ queries to the unknown channel $\cN_{A\to B}$~\cite[Eq.~(336)]{utsumi2025}, where $\lambda_{\min}$ is the minimum non-zero eigenvalue of $\cN(\sigma)$ and $\tau_{\min}$ is the minimum non-zero eigenvalue of the Choi state $\sfN_{AB}/d_A$. As noted in Ref.~\cite{utsumi2025}, implementing steps i) and iii) of the Petz map additionally requires access to $\sigma_A$ and $\cN_{A\to B}$, whereas these can be straightforwardly realized using the density matrix exponentiation~\cite{lloyd2014quantum} and the QSVT. Therefore, the query complexity regarding $\cN_{A\to B}$ discussed here is primarily about step ii), i.e., the realization of $\cN^\dag$.

To benchmark our protocol against this approach for the task of expectation value estimation, we consider an estimator $\hat{\xi}$ constructed by sampling from their approximate channel, denoted as $\cP^{\varepsilon_0}_{\sigma,\cN}$. For an observable $O_A$ with $\|O_A\|_{\infty}\le 1$, obtaining an estimate with precision $\varepsilon$ and confidence $1-\delta$ requires setting the internal channel simulation error to $\varepsilon_0 = \varepsilon/2$ to bound the bias, while using $O\big(\frac{1}{\varepsilon^2}\log(\frac{1}{\delta})\big)$ samples to suppress statistical fluctuation. 
This error budget is justified as follows. First, notice that 
\begin{equation*}
\begin{aligned}
\Big|\!\tr\!\big[O_A\cP_{\sigma,\cN}(\omega_B)\big]-\hat{\xi}\Big| &\leq \Big|\!\tr\!\big[O_A(\cP_{\sigma,\cN}-\cP_{\sigma,\cN}^{\varepsilon/2})(\omega_B)\big]\Big| \\
&\quad + \Big|\! \tr\!\big[O_A\cP_{\sigma,\cN}^{\varepsilon/2}(\omega_B)\big] -\hat{\xi}\Big|\\
&\leq \big\|(\cP_{\sigma,\cN}-\cP_{\sigma,\cN}^{\varepsilon/2})(\omega_B)\big\|_1\\
&\quad + \Big|\! \tr\!\big[O_A\cP_{\sigma,\cN}^{\varepsilon/2}(\omega_B)\big] -\hat{\xi}\Big|,
\end{aligned}
\end{equation*}
where we have used the triangle inequality and the H\"older inequality in the first and the second inequality, respectively. Then, by the inclusion of events, it follows that
\begin{equation}
\begin{aligned}
&\Pr\Big(\big|\!\tr\!\big[O_A\cP_{\sigma,\cN}(\omega_B)\big]-\hat{\xi}\big|\leq \varepsilon\Big)\\
\geq &\Pr\Big(\big|\!\tr\!\big[O_A\cP^{\varepsilon/2}_{\sigma,\cN}(\omega_B)\big]-\hat{\xi}\big|\leq \frac{\varepsilon}{2} \ \cap \\
&\quad\quad\quad \big\|\cP_{\sigma,\cN}(\omega_B)- \cP^{\varepsilon/2}_{\sigma,\cN}(\omega_B)\big\|_1\leq \frac{\varepsilon}{2}\Big)\\
=& \Pr\Big(\big|\!\tr\!\big[O_A\cP^{\varepsilon/2}_{\sigma,\cN}(\omega_B)\big]-\hat{\xi}\big|\leq \frac{\varepsilon}{2}\Big) \geq 1-\delta,
\end{aligned}
\end{equation}
where the first inequality relies on the fact that satisfying both the bias and statistical constraints guarantees the total error bound through triangle inequality, and the equality holds because $\big\|\cP_{\sigma,\cN}(\omega_B)- \cP^{\varepsilon/2}_{\sigma,\cN}(\omega_B)\big\|_1\leq \frac{\varepsilon}{2}$ is satisfied deterministically by the channel construction. Since each sample requires a fresh realization of the adjoint channel, the total query complexity for $\cN_{A\to B}$ via this direct reduction scales as
\begin{equation}\label{eq:scaling_uhlmann}
O\left(\frac{d_A^{11/2}d_B^{5/2}}{\varepsilon^4\lambda_{\min}^{3/2} } \log\left(\frac{1}{\delta}\right)\min{\left\{ \frac{1}{\tau^2_{\min}}, \frac{d_A^8 d_B^4}{\varepsilon^4 \lambda_{\min}^2 }\right\}}\right).
\end{equation} 
We note that this scaling represents an upper bound on the cost, derived from utilizing the channel simulation results of Ref.~\cite{utsumi2025} as a straightforward subroutine for expectation value estimation.

In the following, we analyze the query complexity of the unknown channel $\cN_{A\to B}$ for estimating $\tr[O_A\cP_{\sigma,\cN}(\omega_B)]$ using our method. Notice that our method requires a post-selection arising from the transpose realization, which means we need to reject samples during our process. The success or acceptance across each sampling is independent, so the number of accepted samples $N$ after $M$ attempts follows a binomial distribution $N\sim \cB(M,\eta)$. First, consider the post-selection probability in realizing the channel transpose part. When $(x,y)=(+,+)$, we have
\begin{equation}
p_{\text{suc}} = \frac{\tr\Big[ \frac{1}{d_A+1}\big(\idop + \cN(\frac{\idop}{d_A})^{\pT}\big)K\Big]}{d_B\cdot\tr K} \geq \frac{1}{d_B(d_A+1)},
\end{equation}
where $K \coloneqq \cN(\sigma)^{-1/2}\omega\cN(\sigma)^{-1/2}$. When $(x,y)=(+,-)$ or $(-,-)$, we have
\begin{equation}
\begin{aligned}
p_{\text{suc}} &= \frac{\tr\Big[ \frac{1}{d_A-1}\big(\idop - \cN(\frac{\idop}{d_A}\big)^{\pT}\big)K\Big]}{d_B\cdot\tr K}\\
&=\frac{1}{d_B(d_A-1)} - \frac{1}{d_Ad_B(d_A-1)} \frac{\tr [\cN(\idop)^{\pT} K]}{\tr K}\\
&\geq \frac{1}{d_B(d_A-1)}(1-\zeta_{\max}),
\end{aligned} 
\end{equation}
where $\zeta_{\max}$ is the maximum eigenvalue of $\cN(\frac{\idop}{d_A})$.
Denote
\begin{equation}
    \eta \coloneqq \min\Big\{\frac{1}{d_B(d_A+1)},\frac{1-\zeta_{\max}}{d_B(d_A-1)}\Big\}.
\end{equation}
Now consider event $A$ of fewer than $n$ accepted samples with $\Pr(A)\leq \delta/2$, and event $B$ of given $n$ accepted samples, the estimate error exceeds $\varepsilon$ with $\Pr(B)\leq \delta/2$. It follows that the overall failure probability $\Pr(A\cup B)\leq \Pr(A)+\Pr(B)\leq \delta$. Let's consider $n$ success samples and let $\Pr(N < n)\leq \delta/2$. The Chernoff bound~\cite{Chernoff1952} gives
\begin{equation}
    \Pr(N\leq (1-\beta)M\eta )\leq \exp\Big(-\frac{\beta^2 M\eta}{2}\Big),~\beta\in(0,1).
\end{equation}
Set $(1-\beta)M\eta = n$, so $\beta = 1-\frac{n}{M\eta}$. We require $\exp\big(-\frac{\beta^2 M\eta}{2}\big) \leq \frac{\delta}{2}$. Taking logarithms and solving the quadratics with respect to $M\eta$ yields 
\begin{equation}
M\eta \geq n+\log(2/\delta) + \sqrt{\log(2/\delta)^2 + 2n\log(2/\delta)}.
\end{equation}
Based on the above, a sample size of
\begin{equation}
\begin{aligned}
    M = O\left(\frac{d_A^2d_B^2}{\eta \varepsilon^2}\log\left(\frac{1}{\delta}\right)\right)
\end{aligned}
\end{equation}
suffices to obtain an $\varepsilon$-close estimate of $\tr[O_A \cP_{\sigma,\cN}(\omega_B)]$ with probability at least $1-\delta$. In particular, when $\cN$ is a unital channel, we have $M = O\left(\frac{d_A^3d_B^3}{\varepsilon^2}\log\left(\frac{1}{\delta}\right)\right)$, which significantly improves Eq.~\eqref{eq:scaling_uhlmann}.

\section{Conclusion}

In this work, we have established a strict hierarchy for the universal physical realizability of dual quantum maps. While the channel transpose admits a probabilistic implementation via teleportation, we proved that neither the complex conjugate nor the adjoint can be realized by any completely positive supermap. We overcame this fundamental obstruction by utilizing virtual combs, which allow for the universal simulation of these maps via quasi-probability sampling. We proved the optimality of our construction with respect to the sampling overhead. As a key application, this framework enables the estimation of the Petz recovery map for unknown channels with a query complexity of $O(\varepsilon^{-2})$, offering an improvement over existing deterministic approximations. We expect our protocol for simulating channel adjoint can lead to more applications in experimental probing of OTOCs in generic open quantum systems~\cite{Li2017,Tuziemski2019,Alonso2019,Blocher2022,Xu2024}, i.e., $\tr[\cN^\dag(W)V^\dag \cN(W)V\rho]$ where $V,W$ are some Hermitian and $\cN$ is the underlying dynamic of the open quantum system. Future work may extend these virtual techniques to multi-slot settings and broader classes of higher-order transformations in quantum thermodynamics~\cite{Stratton2024} and error mitigation~\cite{Minjun2026}.

\section*{Acknowledgments}
C.\ Z.\ thanks Zhenhuan Liu for informing the reference~\cite{Dong_2019}. X.\ W.\ was partially supported by the National Key R\&D Program of China (Grant No.~2024YFB4504004), the National Natural Science Foundation of China (Grant No.~12447107, 92576114), the Guangdong Provincial Quantum Science Strategic Initiative (Grant No.~GDZX2403008, GDZX2503001), the CCF-Tencent Rhino-Bird Open Research Fund, and the Guangdong Provincial Key Lab of Integrated Communication, Sensing and Computation for Ubiquitous Internet of Things (Grant No.~2023B1212010007). 
G.\ B.\ acknowledges support from the Start-up Fund (Grant No.~G0101000274) from The Hong Kong University of Science and Technology (Guangzhou).
Y.\ L.\ was supported in part by the National Nature Science Foundation of China under Grant 62302346; in part by Hubei Provincial Nature Science Foundation of China under Grant 2024AFA045; and in part by the Fundamental Research Funds for the Central Universities under Grant 2042025kf0023.

\bibliography{arxiv_v1}

\appendix

\vspace{4cm}
\onecolumngrid

\begin{center}
\large{\textbf{Appendix}}
\end{center}


\section{Proof of Lemma~\ref{lemma:basenorm=diamondnorm} and Theorem~\ref{thm:min_basenorm}}\label{app:proof}

\renewcommand\theproposition{\ref{lemma:basenorm=diamondnorm}}
\setcounter{proposition}{\arabic{proposition}-1}
\begin{lemma}
    For any $n$-slot virtual comb $\widetilde{\cC}$, it holds that $\|\widetilde{\cC}\|_{\bdia} = \|\widetilde{\cC}\|_{\dia}$.
\end{lemma}

\begin{proof}
This generalizes the result for HPTP linear maps given in~\cite[Theorem 3]{regula2021operational}, and the proof directly utilizes the SDP formulation of two norms. 
Let us use the notation $_{B}\sfX_{AB} \coloneqq \tr_B \sfX_{AB} \ox \frac{\idop_B}{d_B}$. Note that the diamond norm of an $n$-slot virtual comb $\widetilde{\sfC}$ can be formulated as the following SDP problem~\cite[Eq.~(19)]{Chiribella_2008m}: 
\begin{equation}
\begin{aligned}\label{eq:supermap_dianorm_seq}
 \|\widetilde{\cC}\|_\dia = \max &\; \big\langle \sfT^{(0)}-\sfT^{(1)},  \widetilde{\sfC}\big\rangle\\
{\rm s.t.} &\;\; \sfT^{(0)},\sfT^{(1)}  \geq 0,\\
&\;\; \sfT=\sfT^{(0)}+\sfT^{(1)},\\
&\; \tr \sfT = d_{I_1}d_{I_2}\cdots d_{F},~\sfT = _{F}\!\!\sfT, \\
&\;\; _{O_{n-k}\dots I_nO_nF}\sfT = _{I_{n-k}O_{n-k}\dots I_nO_nF}\!\!\sfT,~\forall k\in\{ 1,\dots,n-1\}.
\end{aligned}
\end{equation}
It follows that the dual of \eqref{eq:supermap_dianorm_seq} is given by (cf.~\cite[Appendix A.3]{gutoski2012measure})
\begin{equation}
\begin{aligned}\label{eq:supermap_dianorm_dual_seq}
\|\widetilde{\cC}\|_\dia = \min &\;\; p\\
{\rm s.t.} &\; -\sfS\leq \widetilde{\sfC} \leq \sfS,~ _{F}\sfS = _{O_nF}\!\sfS, \\
&\;\;  _{O_{n-k-1}I_{n-k}\dots I_nO_nF}\sfS = _{I_{n-k}\dots I_nO_nF}\!\sfS,~\forall k \in\{1,\dots,n-1\},\\
&\;\; \tr\sfS = p d_{PI_1O_1\dots I_nO_n},
\end{aligned}
\end{equation}
and the strong duality holds.
Meanwhile, we write the base norm explicitly as
\begin{equation}\label{eq:sdp_sample_seq}
\begin{aligned}
\|\widetilde{\cC}\|_{\bdia}= \min &\;\; \alpha_++\alpha_- \\
{\rm s.t.} &\;\; \sfC_{+},\sfC_{-} \geq 0,\\
&\;\; \widetilde{\sfC} = \sfC_{+} - \sfC_{-},\\
&\;\; \forall k \in\{1,\dots,n-1\}, x \in\{+,-\}:\\
&\;\; _{F}\sfC_{x} = _{O_nF}\sfC_{x}, \\
&\;\;  _{O_{n-k-1}I_{n-k}\dots I_nO_nF}\sfC_{x} = _{I_{n-k}\dots I_nO_nF}\sfC_{x},\\
&\; \tr_{I_1 \cdots I_nF} \sfC_{x} = \alpha_x\idop_{PO_1\cdots O_n}.\\
\end{aligned}
\end{equation}
For any feasible pair $(p,\sfS)$ of~\eqref{eq:supermap_dianorm_dual_seq}, we note that $p\geq 1$ because $\sfS \geq \widetilde{\sfC}$ gives $\tr \sfS \geq \tr\widetilde{\sfC} = d_{PI_1O_1\dots I_nO_n}$. 
Then we can construct $(\alpha_+,\alpha_-,\sfC_{+},\sfC_{-}) \coloneqq (\frac{p+1}{2}, \frac{p-1}{2}, \frac{\sfS+\widetilde{\sfC}}{2},\frac{\sfS-\widetilde{\sfC}}{2})$ which can be checked to be a feasible tuple of~\eqref{eq:sdp_sample_seq}. Conversely, for any feasible tuple $(\alpha_+,\alpha_-,\sfC_{+},\sfC_{-})$ of \eqref{eq:sdp_sample_seq}, we have that the pair $(p,\sfS) \coloneqq (\alpha_++\alpha_-,\sfC_{+}+\sfC_{-})$ is a feasible pair of ~\eqref{eq:supermap_dianorm_dual_seq}. Since the objective functions are the same for both problems, we can conclude that \eqref{eq:supermap_dianorm_dual_seq} and \eqref{eq:sdp_sample_seq} have the same optimal value, and hence $\|\widetilde{\cC}\|_{\bdia} = \|\widetilde{\cC}\|_{\dia}$.
\end{proof}
\renewcommand{\theproposition}{S\arabic{proposition}}

\renewcommand\theproposition{\ref{thm:min_basenorm}}
\setcounter{proposition}{\arabic{proposition}-1}
\begin{theorem}
The minimum base norm of a 1-slot virtual comb that satisfies $\widetilde{\cC}(\cN)=\cN^*,~\forall \cN\in\CPTP(A,B)$ is $d_Ad_B-d_A+1$.
\end{theorem}
\begin{proof}
The minimum base norm of a 1-slot virtual comb that can realize the complex conjugate for any quantum channel is given by the following SDP.
\begin{align}
    \nu(d_A,d_B) \coloneqq \min &\;\; p_1 + p_2\notag \\
    {\rm s.t.} &\;\; \sfC^{(1)}_{A'ABB'} \geq 0, ~\sfC^{(2)}_{A'ABB'} \geq 0,\\
    &\;\; \tr_{B'} \sfC^{(i)}_{A'ABB'} = \tr_{BB'} \sfC^{(i)}_{A'ABB'} \ox \frac{\idop_{B}}{d_{B}},\\
    &\;\; \tr_{AB'} \sfC^{(i)}_{A'ABB'} = p_i\idop_{A'B},~\forall i\in\{1,2\},\\
    &\;\; \tr_{AB} \Big[(\sfC^{(1)}_{A'ABB'}-\sfC_{A'ABB'}^{(2)})(\sfN_{AB}^{\pT}\ox \idop_{A'B'})\Big] = \sfN_{AB}^{\pT},~ \forall \cN_{A\to B}.\label{sdpcond:transpose}
\end{align}
Define $\mathscr{T}_{AB}\coloneqq \{X\in\mathscr{L}^{\dag}_{AB}:\tr_B X =\idop_A\}$ as an affine subspace of $\mathscr{L}^{\dag}_{AB}$, the space of Hermitian operators on system $AB$. We have that 
\begin{equation}\label{Eq:TP_equiv}
\mathscr{T}_{AB} = \Big\{\frac{1}{d_B}\idop_{AB} + \Delta:\Delta\in\mathscr{L}^{\dag}_{AB}, \tr_B\Delta = 0\Big\} = \frac{1}{d_B}\idop_{AB} + \ker(\tr_B).    
\end{equation} 
Noticing that $\dim\mathscr{L}^{\dag}_{AB}= d_A^2 d_B^2$ and $\mathrm{rank}(\tr_B) = d_A^2$, we have that $\dim \ker(\tr_B)=d_A^2(d_B^2-1)$ and let $\{\Delta_j\}_{j=1}^{d_A^2(d_B^2-1)}$ be a basis of $\ker(\tr_B)$. It follows that for any Choi operator $\sfN_{AB}$ of a quantum channel $\cN_{A\to B}$, it can be written as $\sfN_{AB}=\idop_{AB}/d_B + \sum_{j=1}^{d_A^2(d_B^2-1)} \alpha_j \Delta_j$. Now denote the condition in Eq.~\eqref{sdpcond:transpose} as
\begin{equation}
f(\cdot) = \tr_{AB} \Big[(\sfC^{(1)}_{A'ABB'}-\sfC_{A'ABB'}^{(2)})((\cdot)^{\pT}\ox \idop_{A'B'})\Big] - (\cdot)^{\pT}.
\end{equation}
Suppose it satisfies that $f(\Delta_j) = 0,~\forall j\in\{1,2,\dots,d_A^2(d_B^2-1)\}$ and $f(\idop_{AB}/d_B)=0$. By the linearity of $f(\cdot)$, it follows that for any Choi operator $\sfN_{AB}$ of a quantum channel $\cN_{A\to B}$, we have that $f(\sfN_{AB})=f\big(\idop_{AB}/d_B + \sum_{j} \alpha_j \Delta_j\big) = 0$. On the other hand, suppose it satisfies that $f(\sfN_{AB})=0,~\forall \sfN_{AB}\in\mathscr{T}_{AB}\cap\mathscr{P}_{AB}$. Noticing that $\idop_{AB}/d_B, \idop_{AB}/d_B \pm \ve \Delta_j \in\mathscr{T}_{AB}\cap\mathscr{P}_{AB}$ for sufficiently small $\ve$ and any $\Delta_j$, we have that $f(\idop_{AB}/d_B) = 0$ and $f(\Delta_j) = 0,\forall j\in\{1,2,\dots,d_A^2(d_B^2-1)\}$. Therefore, we conclude that the condition of $f(\sfN_{AB}) = 0,\forall \cN_{A\to B}$ is equivalent to 
\begin{equation}
f(\sfD_{AB}) = 0 \quad \text{and} \quad f(\Delta_j=0),\forall j\in\{1,2,\dots,d_A^2(d_B^2-1)\},    
\end{equation}
where $\cD$ is the fully depolarizing channel. We can rewrite 
\begin{equation}\label{sdp:basis_primal}
\begin{aligned}
    \nu(d_A,d_B) = \min &\;\; p_1 + p_2 \\
    {\rm s.t.} &\;\; \sfC^{(1)}_{A'ABB'} \geq 0, ~\sfC^{(2)}_{A'ABB'} \geq 0,\\
    &\;\; \tr_{B'} \sfC^{(i)}_{A'ABB'} = \tr_{BB'} \sfC^{(i)}_{A'ABB'} \ox \frac{\idop_{B}}{d_{B}},\\
    &\;\; \tr_{AB'} \sfC^{(i)}_{A'ABB'} = p_i\idop_{A'B},~\forall i\in\{1,2\},\\
    &\;\; \tr_{AB} \Big[(\sfC^{(1)}_{A'ABB'}-\sfC_{A'ABB'}^{(2)})(\sfD_{AB}^{\pT}\ox \idop_{A'B'})\Big] = \sfD_{AB}^{\pT},\\
    &\;\; \tr_{AB} \Big[(\sfC^{(1)}_{A'ABB'}-\sfC_{A'ABB'}^{(2)})(\Delta_j^{\pT}\ox \idop_{A'B'})\Big] = \Delta_j^{\pT},~j\in\{1,2,\dots,d_A^2(d_B^2-1)\}.
\end{aligned}
\end{equation}
According to Eq.~\eqref{Eq:Choi_C123} in Theorem~\ref{thm:conj_HPTP}, by choosing
\begin{equation}\label{eq:primal_optimal}
\begin{aligned}
    & \widetilde{\sfC}^{(1)}_{A'ABB'} \coloneqq \frac{d_B+1}{2}\sfW_{d_A}^{+} \ox \sfW_{d_B}^{+} + \frac{(d_A-1)(d_B-1)}{2} \sfW_{d_A}^{-} \ox \sfW_{d_B}^{-},\\
    & \widetilde{C}^{(2)}_{A'ABB'} \coloneqq \frac{d_A(d_B-1)}{2} \sfW_{d_A}^{+} \ox \sfW_{d_B}^{-},\\
    & \hat{p}_1 \coloneqq \frac{d_Ad_B - d_A + 2}{2},\\
    & \hat{p}_2 \coloneqq \frac{d_A(d_B-1)}{2},
\end{aligned}
\end{equation}
we have that $\{\widetilde{\sfC}^{(1)}_{A'ABB'},\widetilde{\sfC}^{(2)}_{A'ABB'},\hat{p}_1,\hat{p}_2\}$ is a feasible solution to the SDP~\eqref{sdp:basis_primal}. Therefore, we conclude that 
\begin{equation}\label{Eq:nu_leq}
    \nu(d_A,d_B) \leq \frac{d_Ad_B - d_A + 2}{2} + \frac{d_A(d_B-1)}{2} = d_Ad_B - d_A + 1.
\end{equation}
In the following, we will use the SDP duality to prove the converse. The Lagrangian of the problem~\eqref{sdp:basis_primal} is given by
\begin{equation}
\begin{aligned}
\mathcal{L} = &\ p_1 + p_2 + \sum_{i=1}^2 \left\langle Y_{A'AB}^{(i)},\ \tr_{B'} \sfC^{(i)} - \tr_{BB'} \sfC^{(i)} \ox \frac{\idop_{B}}{d_{B}} \right\rangle \\
&+ \sum_{i=1}^2 \left\langle Z_{A'B}^{(i)},\ \tr_{AB'} \sfC^{(i)} - p_i \idop_{A'B} \right\rangle \\
&+ \left\langle X_{A'B'},\ \sfD_{AB}^\pT - \tr_{AB}\left[(\sfC^{(1)} - \sfC^{(2)})(\sfD_{AB}^\pT \ox \idop_{A'B'}) \right] \right\rangle \\
&+ \sum_{j=1}^{d_A^2(d_B^2-1)} \left\langle G_j,\ \Delta_j^{\pT} - \tr_{AB}\left[(\sfC^{(1)} - \sfC^{(2)})(\Delta_j^\pT \ox \idop_{A'B'}) \right] \right\rangle\\
=&\ \tr[X\sfD_{AB}^\pT] + \sum_{j=1}^{d_A^2(d_B^2-1)} \tr[G_j\Delta_j^\pT] + p_1 (1-\tr Z^{(1)}) + p_2 (1-\tr Z^{(2)}) \\
&+ \Big\langle \sfC^{(1)},\ Y^{(1)}_{A'AB}\ox \idop_{B'} -  \frac{1}{d_{B}}Y^{(1)}_{A'A}\ox \idop_{BB'} + Z^{(1)}_{A'B}\ox \idop_{AB'} - \sfD_{AB}^\pT \ox X_{A'B'} - \sum_j\Delta_j^\pT \ox G_j\Big\rangle\\
&+ \Big\langle \sfC^{(2)},\ Y^{(2)}_{A'AB}\ox \idop_{B'} -  \frac{1}{d_{B}}Y^{(2)}_{A'A}\ox \idop_{BB'} + Z^{(2)}_{A'B}\ox \idop_{AB'} + \sfD_{AB}^\pT \ox X_{A'B'} + \sum_j\Delta_j^\pT \ox G_j\Big\rangle,
\end{aligned}
\end{equation}
where $Y^{(1)}_{A'AB},Y^{(2)}_{A'AB}\in \mathscr{L}^{\dag}_{A'AB},~Z_{A'B}^{(1)},Z_{A'B}^{(2)}\in\mathscr{L}^{\dag}_{A'B},X_{A'B'},G_j\in\mathscr{L}^\dag_{A'B'}$ are dual variables. It follows that the dual problem is given by
\begin{equation}\label{sdp:basis_dual}
\begin{aligned}
    \max &\;\; \tr\left[X_{A'B'}\sfD_{AB}^\pT\right] + \sum_{j=1}^{d_A^2(d_B^2-1)} \tr\left[G_j\Delta_j^\pT\right] \\
    {\rm s.t.} &\; \tr Z_{A'B}^{(1)} = 1, ~\tr Z_{A'B}^{(2)} = 1,\\
    &\;\; Y^{(1)}_{A'AB}\ox \idop_{B'} - \frac{1}{d_{B}}Y^{(1)}_{A'A}\ox \idop_{BB'} + Z^{(1)}_{A'B}\ox \idop_{AB'}\geq \sfD_{AB}^\pT \ox X_{A'B'} + \sum_j\Delta_j^\pT \ox G_j,\\
    &\;\;  -Y^{(2)}_{A'AB}\ox \idop_{B'} +  \frac{1}{d_{B}}Y^{(2)}_{A'A}\ox \idop_{BB'} - Z^{(2)}_{A'B}\ox \idop_{AB'}\leq  \sfD_{AB}^\pT \ox X_{A'B'} +\sum_j\Delta_j^\pT \ox G_j.
\end{aligned}
\end{equation}
For the primal problem, let
\begin{equation}
    \begin{aligned}
         \bar{\sfC}^{(1)}_{A'ABB'} &= F_{A'A}\ox F_{BB'} + \epsilon \idop_{A'ABB'} > 0,\\
         \bar{\sfC}^{(2)}_{A'ABB'} &= \epsilon \idop_{A'ABB'} >0, \text{ for }\epsilon > 1,\\
    \end{aligned}
\end{equation}
then $\{\bar{\sfC}^{(1)}_{A'ABB'},\bar{\sfC}^{(2)}_{A'ABB'},\bar{p}_1 = 1+\epsilon d_Ad_B, \bar{p}_2 = \epsilon d_Ad_B\}$ is a strictly feasible solution. This indicates that the strong duality holds by Slater's condition~\cite{ponstein2004approaches}. For the basis $\{\Delta_j\}_{j=1}^{d_A^2(d_B^2-1)}$ on system $A$ and $B$, we may choose the generalized Gell-Mann matrices~\cite{bertlmann2008bloch} as generators, given by
\begin{equation}
\begin{aligned}
    &\Lambda_0 \coloneqq \frac{\idop_d}{\sqrt{d}}, \\
    &\Lambda_{m} \coloneqq \sqrt{\frac{1}{m(m+1)}} \left(\sum_{j=1}^m \ketbra{j}{j} - m\ketbra{m+1}{m+1}\right), \quad 1\leq m  \leq d-1, \\
    &\Lambda^s_{jk} \coloneqq \frac{1}{\sqrt 2}(\ketbra{j}{k} - \ketbra{k}{j}) , \quad\quad~ 1\leq j \leq k \leq d, \\
    &\Lambda^a_{jk} \coloneqq \frac{1}{\sqrt 2}(-i\ketbra{j}{k} +i \ketbra{k}{j}) , \quad 1\leq j \leq k \leq d. 
\end{aligned}
\end{equation}
Denote 
\begin{equation}
    \big\{L_d^{(n)}\big\}_{n=0}^{d^2-1} = \{\Lambda_m\}_{m=0}^{d-1} \cup \{\Lambda^s_{jk}\}_{j<k} \cup \{\Lambda^a_{jk}\}_{j<k}
\end{equation}
with $L_d^{(0)} = L_0 $. Then we choose $\Delta_j$ to be
\begin{equation}
    \big\{\Delta_j\big\}_{j=1}^{d_A^2(d_B^2-1)} \coloneqq \left\{L_{A}^{(\alpha)} \ox L_{B}^{(\beta)}: 0\leq \alpha \leq d_A^2-1, 1\leq \beta \leq d_B^2-1 \right\}.
\end{equation}
Note that for $\{\Delta_j\}$, they are i). Hermitian: $\Delta_j = \Delta_j^\dagger$, ii). orthonormal: $\tr[\Delta_i^\dagger \Delta_j] = \delta_{ij}$, iii). $\tr_A[\Delta_j] = 0$ except $L_{A}^{(\alpha)} = L_{A}^{(0)}$.
Based on this basis, we construct the following operators to form a feasible solution of the dual SDP~\eqref{sdp:basis_dual}:
\begin{equation}
\label{eq:feasible_sol_dual}
    \begin{aligned}
    X_{A'B'} &= \frac{\idop _{A'B'}}{d_{A'}d_{B'}} , \\
    G_j &= \begin{cases}
        \frac{d_A+d_B}{d_A(d_B^2+d_B)} \Delta_j \ & \text{   if } \tr_A[\Delta_j] \neq 0,\\
        \frac{1}{d_A(d_B+1)}  \Delta_j  & \text{  otherwise},
    \end{cases} \\
    Z^{(1)}_{A'B} &= Z^{(2)}_{A'B} = \frac{1}{d_{A'}d_{B}}\idop_{A'B} ,\\ 
    Y^{(1)}_{A'AB} &= Y^{(2)}_{A'AB} =  \idop_{A'AB}.
\end{aligned}
\end{equation}
The constraints are checked to be satisfied by the following calculation. Write the summation explicitly, we have that
\begin{equation}
\begin{aligned}
&\sum_j\Delta_j^\pT \ox G_j \\
= &\;    \frac{d_A + d_B}{d_A(d_B^2 + d_B)} \frac{F_{A'A}}{d_A} \ox \Big( F_{BB'} - \frac{\idop_{BB'}}{d_B} \Big) + \frac{1}{d_A(d_B + 1)}  \left[ F_{A'A} \ox \Big( F_{BB'} - \frac{\idop_{BB'}}{d_B} \Big) - \frac{F_{A'A}}{d_A} \ox \Big( F_{BB'} - \frac{\idop_{BB'}}{d_B} \Big) \right] \\
=&\;\frac{1}{d_A(d_B+1)} \Big( F_{A'A} + \frac{\idop_{A'A}}{d_B} \Big) \ox \Big( F_{BB'} - \frac{\idop_{BB'}}{d_B} \Big).
\end{aligned}
\end{equation}
The eigenvalues of $\sum_j\Delta_j^\pT \ox G_j$ are 
\begin{equation}
\begin{aligned}
\lambda_1 &= \frac{1}{d_A(d_B+1)} \left(1 - \frac{1}{d_B^2}\right) = \frac{d_B - 1}{d_A d_B^2}, \\
\lambda_2& = \frac{-1}{d_A(d_B+1)}\left(1 + \frac{1}{d_B}\right)^2 = -\frac{d_B + 1}{d_A d_B^2}, \\
\lambda_3& = \frac{-1}{d_A(d_B+1)}\left(1 - \frac{1}{d_B}\right)^2.
\end{aligned}
\end{equation}
Clearly, $\lambda_1$ is the largest eigenvalue and $\lambda_2$ is the smallest eigenvalue. Hence, we conclude that 
\begin{equation}
    \frac{d_B - 1}{d_A d_B^2}\idop_{A'ABB'} \geq \sum_j\Delta_j^\pT \ox G_j \geq -\frac{d_B + 1}{d_A d_B^2}\idop_{A'ABB'}.
\end{equation}
It follows that \eqref{eq:feasible_sol_dual} is a feasible solution since the constraint can be expressed as
\begin{equation}
\begin{aligned}
& Y^{(1)}_{A'AB}\ox \idop_{B'} - \frac{1}{d_{B}}Y^{(1)}_{A'A}\ox \idop_{BB'} + Z^{(1)}_{A'B}\ox \idop_{AB'}-\sfD_{AB}^\pT \ox X_{A'B'} \\
=&\;  0+ \left(\frac{1}{d_A d_B} -\frac{1}{d_A d_B^2} \right)\idop_{A'ABB'} = \frac{d_B - 1}{d_A d_B^2}\idop_{A'ABB'}\geq\sum_j\Delta_j^\pT \ox G_j,
\end{aligned}
\end{equation}
and 
\begin{equation}
\begin{aligned}
& -Y^{(2)}_{A'AB}\ox \idop_{B'} + \frac{1}{d_{B}}Y^{(2)}_{A'A}\ox \idop_{BB'} - Z^{(2)}_{A'B}\ox \idop_{AB'}-\sfD_{AB}^\pT \ox X_{A'B'} \\
=&\;  0+ \left(-\frac{1}{d_A d_B} -\frac{1}{d_A d_B^2} \right)\idop_{A'ABB'} = -\frac{d_B +1}{d_A d_B^2}\idop_{A'ABB'}\leq\sum_j\Delta_j^\pT \ox G_j.
\end{aligned}
\end{equation}
For the feasible solution given in Eq.~\eqref{eq:feasible_sol_dual}, the objective value is
\begin{equation}
\begin{aligned}
&\tr\left[X_{A'B'}\sfD_{AB}^\pT\right] + \sum_{j=1}^{d_A^2(d_B^2-1)} \tr\left[G_j\Delta_j^\pT\right] \\
=&\; \frac{1}{d_B}+\frac{d_A+d_B}{d_A(d_B^2+d_B)} (d_B^2-1)  + \frac{1}{d_A(d_B+1)} \left[d_A^2(d_B^2-1) - (d_B^2-1)\right] \\
=&\; \frac{1}{d_B}+\frac{(d_A + d_B)(d_B-1)}{d_Ad_B} + \frac{(d_A^2-1)(d_B-1)}{d_A} \\
=&\; \frac{1}{d_B}+ \frac{d_Ad_B-d_A+d_B^2-d_B + d_A^2d_B^2-d_A^2d_B-d_B^2+d_B}{d_Ad_B} \\
=&\; \frac{1}{d_B}+1-\frac{1}{d_B} +d_Ad_B - d_A\\
=&\; d_Ad_B - d_A +1.
\end{aligned}
\end{equation}
Since the dual problem is a maximization problem, we therefore conclude that 
\begin{equation}\label{Eq:nu_geq}
    \nu(d_A,d_B) \geq d_Ad_B - d_A +1.
\end{equation}
Combining Eq.~\eqref{Eq:nu_leq} and Eq.~\eqref{Eq:nu_geq}, we complete the proof.
\end{proof}
\renewcommand{\theproposition}{S\arabic{proposition}}

\end{document}